\documentclass{article}

\usepackage{graphicx}
\usepackage{latexsym}
\usepackage{amsmath}
\usepackage{amsfonts}
\usepackage{amsthm}
\usepackage{url}
\usepackage[linesnumbered,ruled,vlined]{algorithm2e}
\usepackage{setspace}
\usepackage{multirow}
\usepackage{lscape}
\usepackage{rotating}
\usepackage{authblk}

\usepackage{geometry}
\newtheorem{observation}{Observation}
\newtheorem{lemma}{Lemma}

\def\Underline{\setbox0\hbox\bgroup\let\\\endUnderline}
\def\endUnderline{\vphantom{y}\egroup\smash{\underline{\box0}}\\}
\def\|{\verb|}

\newcommand{\ZA}{Z_{\mathcal{A}}}
\newcommand{\ZB}{Z_{\mathcal{B}}}
\newcommand{\ZS}{Z_{\mathcal{S}}}
\newcommand{\Spm}{\mathcal{S}^{\pm}}
\newcommand{\Sup}{\mathcal{S}^{\uparrow}}
\newcommand{\TS}{T_{\Spm}}
\newcommand{\Zup}{Z_{\Sup}}

\title{Enumerating Graph Partitions Without Too Small Connected Components Using Zero-suppressed Binary and Ternary Decision Diagrams}
\date{\vspace{-10mm}}
\author[1]{Yu Nakahata\thanks{nakahata.yu.nm2@is.naist.jp}}
\author[1]{Jun Kawahara\thanks{jkawahara@is.naist.jp}}
\author[1]{Shoji Kasahara\thanks{kasahara@is.naist.jp}}
\affil[1]{Nara Institute of Science and Technology, Ikoma, Japan}

\begin{document}

\maketitle

\begin{abstract}
Partitioning a graph into balanced components is important for several applications.
For multi-objective problems, it is useful not only to find one solution but also to enumerate all the solutions with good values of objectives.
However, there are a vast number of graph partitions in a graph, and thus it is difficult to enumerate desired graph partitions efficiently.
In this paper, an algorithm to enumerate all the graph partitions such that all the weights of the connected components are at least a specified value is proposed.
To deal with a large search space, we use zero-suppressed binary decision diagrams (ZDDs) to represent sets of graph partitions and we design a new algorithm based on frontier-based search, which is a framework to directly construct a ZDD.
Our algorithm utilizes not only ZDDs but also ternary decision diagrams (TDDs) and realizes an operation which seems difficult to be designed only by ZDDs.
Experimental results show that the proposed algorithm runs up to tens of times faster than an existing state-of-the-art algorithm.
\end{abstract}

\section{Introduction}

Partitioning a graph is a fundamental problem in computer science and has several important applications such as evacuation planning, political redistricting, VLSI design, and so on.
In some applications among them, it is often required to balance the weights of connected components in a partition.
For example, the task of the evacuation planning is to design which evacuation shelter inhabitants escape to. This problem is formulated as a graph partitioning problem, and it is important to obtain a graph partition consisting of balanced connected components (each of which contains a shelter and satisfies some conditions). Another example is political redistricting, the purpose of which is to divide a region (such as a prefecture) into several balanced political districts for fairness.

There are a vast number of studies for graph optimization problems. An approach is to use a \emph{zero-suppressed binary decision diagram} (ZDD)~\cite{Minato1993}, which has originally been proposed as a compressed representation of a family of sets.
A distinguished characteristic of the approach is not only to compute the single optimal solution but also to \emph{enumerate} all the feasible solutions in the form of a ZDD.
In addition, using several queries for a family of sets provided by ZDDs, we can impose various constraint conditions on solutions represented by a ZDD.
Using this approach,
Inoue et al.~\cite{Inoue2014} designed an algorithm that constructs the ZDD representing the set of rooted spanning forests
and utilized it to minimize the loss of electricity in an electrical distribution network under complex conditions, e.g., voltage, electric current and phase.
There are other applications such as solving a variant of the longest path problem~\cite{kawahara2017solving}, reliability evaluation~\cite{Hardy2007,ISI1999}, some puzzle problems~\cite{Yoshinaka2012}, and exact calculation of impact diffusion in Web~\cite{Maehara2017}.

For balanced graph partitioning, Kawahara et al.~\cite{kawahara2017generating} proposed an algorithm to construct a ZDD representing the set of balanced graph partitions by frontier-based search~\cite{kawahara2017frontier,Knuth11,Sekine1995}, which is a framework to directly construct a ZDD, and applied it to political redistricting. However, their method stores the weights of connected components, represented as integers, into the ZDD, which generates a not compressed ZDD. As a result, the computation is tractable only for graphs only with less than 100 vertices. Nakahata et al.~\cite{nakahata2018} proposed an algorithm to construct the ZDD representing the set of partitions such that all the weights of connected components are bounded by a given upper threshold (and applied it to evacuation planning). Their approach enumerates connected components with weight more than the upper threshold as a ZDD, say \emph{forbidden components}, and constructs a ZDD representing partitions not containing any forbidden component \emph{as a subgraph} by set operations, which are performed by so-called apply-like methods~\cite{bryant1986graph}. However, it seems difficult to directly use their method to obtain balanced partitions by letting connected components with weight less than a lower threshold be forbidden components because partitions not containing any forbidden component \emph{as a connected component} (i.e., one of parts in a partition coincides a forbidden component) cannot be obtained by apply-like methods.

In this paper, for a ZDD $\ZA$ and an integer $L$, we propose a novel algorithm to construct the ZDD representing the set of graph partitions such that the partitions are represented by $\ZA$ and all the weights of the connected components in the partitions are at least $L$.
The input ZDD $\ZA$ can be the sets of spanning forests used for evacuation planning (e.g., \cite{nakahata2018}), rooted spanning forests used for power distribution networks (e.g., \cite{Inoue2014}), and simply connected components representing regions (e.g., \cite{kawahara2017generating}), all of which satisfy complex conditions according to problems. We generically call these structures ``partitions.''
Roughly speaking, our algorithm excludes partitions containing any forbidden component as a connected component from $\ZA$. We first construct the ZDD, say $\ZS$, representing the set of forbidden components, each of which has weight less than $L$.
Then, for a component in $\ZS$, we consider the cutset that separates the input graph into the component and the rest.
We represent the set of pairs of every component in $\ZS$ and its cutset as a \emph{ternary decision diagram} (TDD)~\cite{yasuoka1995new}, say $\TS$.
We propose a method to construct the TDD $\TS$ from $\ZS$ by frontier-based search. By using the TDD $\TS$, we show how to obtain partitions each of which belongs to $\ZA$, contains all the edges in a component of a pair in $\TS$ and contains no edge in the cutset of the pair. Finally, we exclude such partitions from $\ZA$ and obtain the desired partitions. By numerical experiments, we show that the proposed algorithm runs up to tens of times faster than an existing state-of-the-art algorithm.

This paper is organized as follows.
In Sec.~\ref{sec:preliminaries}, we give some preliminaries and explain ZDDs, TDDs, and frontier-based search.
We describe an overview of our algorithm in Sec.~\ref{sec:overview},
and the detail in the rest of Sec.~\ref{sec:algorithms}.
Section~\ref{sec:experiment} gives experimental results.
The conclusion is described in Sec.~\ref{sec:conclusion}.

\section{Preliminaries}\label{sec:preliminaries}
\subsection{Notation}\label{sec:notation}
Let $\mathbb{Z}^{+}$ be the set of positive integers.
For $k \in \mathbb{Z}^{+}$, we define $[k] = \{1, 2, \dots, k\}$.
In this paper, we deal with a vertex-weighted undirected graph $G = (V, E, p)$,
where $V = [n]$ is the vertex set and $E = \{e_1, e_2, \dots, e_m\} \subseteq \{\{u, v\} \mid u, v \in V\}$ is the edge set.
The function $p \colon V \rightarrow \mathbb{Z}^{+}$ gives the weights of the vertices.
We often drop $p$ from $(V, E, p)$ when there is no ambiguity.
For an edge set $E' \subseteq E$, we call the subgraph $(V, E')$ a \emph{graph partition}.
We often identify the edge set $E'$ with the partition $(V, E')$ by fixing the graph $G$.
For edge sets $E', E''$ with $E'' \subseteq E' \subseteq E$ and
a vertex set $V'' \subseteq V$,
we say that $(V'', E'')$ is included in the partition $(V, E')$ \emph{as a subgraph}.
The subgraph $(V'', E'')$ is called a \emph{connected component} in the partition $(V, E')$ if $V'' = \mathrm{dom}(E'')$ holds, there is no edge in $E' \setminus E''$ incident with a vertex in $V''$, and for any two distinct vertices $u, v \in V''$, there is a $u$-$v$ path on $(V'', E'')$, where $\mathrm{dom}(E'')$ is the set of vertices which are endpoints of at least one edge in $E''$.
In this case, we say that $(V'', E'')$ is included in the partition $(V, E')$ \emph{as a connected component}.
We denote the neighborhood of a vertex $v$ in a partition $E' \subseteq E$ by $N(E', v) = \{u \mid \{u, v\} \in E'\}$.
For $i \in [m]$, $E^{\leq i}$ denotes the set of edges whose indices are at most $i$.
We define $E^{<i}$, $E^{\geq i}$ and $E^{>i}$ in the same way.

For a set $U$, let $U^{+} = \{+e \mid e \in U\}, U^{-} = \{-e \mid e \in U\}$ and $U^{\pm} = U^{+} \cup U^{-}$.
A \emph{signed set} is a subset of $U^{\pm}$ such that, for all $e \in U$, the set contains at most one of $+e$ and $-e$.
For example, when $U = [3]$, both $\{+1, -2\}$ and $\{-3\}$ are signed sets but $\{+1, -1, +3\}$ is not.
A \emph{signed family} is a family of signed sets.
In particular, when $U = E$, we sometimes call a signed set a \emph{signed subgraph} and call a signed family a \emph{set of signed subgraphs}.
For a signed set $S^{\pm}$, we define $\mathrm{abs}(S^{\pm}) = \{e \mid (+e \in S^{\pm}) \lor (-e \in S^{\pm})\}$.

\subsection{Zero-suppressed binary decision diagram}\label{sec:zdd}

\begin{figure}[tb]
\begin{center}
\begin{tabular}{cc}
\begin{minipage}[t]{0.45\columnwidth}
\begin{center}
\includegraphics[width=15mm]{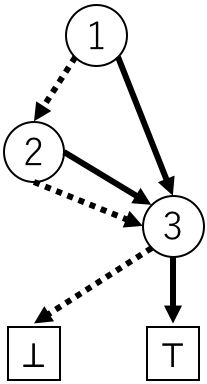}
\caption{The ZDD representing the family $\{\{1, 3\}, \{2, 3\}, \{3\}\}$. A square represents a terminal node. A circle is a non-terminal node and the number in it is a label. A solid arc is a 1-arc and a dashed arc is a 0-arc.}
\label{fig:zdd}
\end{center}
\end{minipage}
\begin{minipage}{0.05\columnwidth}
  \hspace{2mm}
\end{minipage}
\begin{minipage}[t]{0.45\columnwidth}
\begin{center}
\includegraphics[width=15mm]{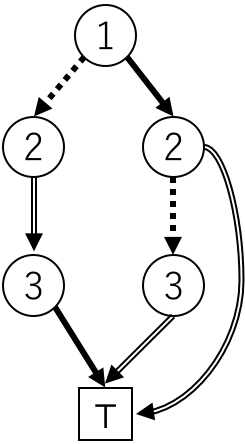}
\caption{The TDD representing the signed family $\{\{+1, -2\}, \{+1, -3\}, \{-2, +3\}\}$. A dashed arc is a ZERO-arc, a solid single arc is a POS-arc and a solid double arc is NEG-arc. For simplicity, $\bot$ and the arcs pointing at it are omitted.}
\label{fig:ztdd}
\end{center}
\end{minipage}
\end{tabular}
\end{center}
\end{figure}

A \emph{zero-suppressed binary decision diagram} (ZDD)~\cite{Minato1993} is a directed acyclic graph $Z = (N_Z, A_Z)$ representing a family of sets.
Here $N_Z$ is the set of \emph{nodes} and $A_Z$ is the set of \emph{arcs}.\footnote{To avoid
confusion, we use the words ``vertex'' and ``edge'' for input graphs and ``nodes'' and ``arcs'' for decision diagrams.}
$N_Z$ contains two \emph{terminal nodes} $\top$ and $\bot$.
The other nodes than the terminal nodes are called \emph{non-terminal nodes}.
Each non-terminal node $\alpha$ has the \emph{0-arc}, the \emph{1-arc}, and the \emph{label} corresponding to an item in the universe set.
For $x \in \{0, 1\}$, we call the destination of the $x$-arc of a non-terminal node $\alpha$ the \emph{$x$-child} of $\alpha$. %, denoted by $c_x(\alpha)$.
We denote the label of $\alpha$ by $l(\alpha)$ and in this paper, assume that
$l(\alpha) \in \mathbb{Z}^{+} \cup \{\infty\}$ for any $\alpha \in N_Z$.
For convenience, we let $l(\top) = l(\bot) = \infty$.
For each directed arc $(\alpha, \beta) \in A_Z$, the inequality $l(\alpha) < l(\beta)$ holds,
which ensures that $Z$ is acyclic.
There is exactly one node whose in-degree is zero,
called the \emph{root node} and denoted by $r_Z$.
The number of the non-terminal nodes of $Z$ is called the \emph{size} of $Z$ and denoted by $|Z|$.

$Z$ represents the family of sets in the following way.
Let $\mathcal{P}_Z$ be the set of all the directed paths from $r_Z$ to $\top$.
For a directed path $p = (n_1, a_1, n_2, a_2, \dots, n_k, a_k, \top) \in \mathcal{P}_Z$ with $n_i \in N_Z$, $a_i \in A_Z$ and $n_1 = r_Z$, we define $S_p = \{l(n_i) \mid a_i \in A_{Z, 1}, i \in [k]\}$, where $A_{Z, 1}$ is the set of the 1-arcs of $Z$.
We interpret that $Z$ represents the family $\{S_p \mid p \in \mathcal{P}_Z\}$.
In other words, a directed path from $r_Z$ to $\top$ corresponds to a set in the family represented by $Z$.
As an example, we illustrate the ZDD representing the family $\{\{1, 3\}, \{2, 3\}, \{3\}\}$ in Fig.~\ref{fig:zdd}.
In the figure, a dashed arc ($\dashrightarrow$) and a solid arc ($\rightarrow$) are a 0-arc and a 1-arc, respectively.
On the ZDD in Fig.~\ref{fig:zdd}, there are three directed paths from the root node to $\top$: $1 \rightarrow 3 \rightarrow \top, 1 \dashrightarrow 2 \rightarrow 3 \rightarrow \top$, and $1 \dashrightarrow 2 \dashrightarrow 3 \rightarrow \top$, which correspond to $\{1, 3\}, \{2, 3\}$, and $\{3\}$, respectively.
We denote a ZDD representing a family $\mathcal{F}$ by $Z_{\mathcal{F}}$.

\subsection{Ternary decision diagram}\label{sec:ztdd}
A \emph{ternary decision diagram} (TDD)~\cite{yasuoka1995new} is a directed acyclic graph $T = (N_T, A_T)$ representing a signed family.
A TDD shares many concepts with a ZDD, and thus we use the same notation as a ZDD for a TDD.
The difference between a ZDD and a TDD is that, while a node of the former has two arcs, that of the latter has three, which are called the \emph{ZERO-arc}, the \emph{POS-arc}, and the \emph{NEG-arc}.

$T$ represents the signed family in the following way.
For a directed path \linebreak
$p = (n_1, a_1, n_2, a_2, \dots, n_k, a_k, \top) \in \mathcal{P}_T$ with $n_i \in N_Z$, $a_i \in A_T$ and $n_1 = r_T$, we define $S^{\pm}_p = \{+l(n_i) \mid a_i \in A_{T, +}, i \in [k] \} \cup \{-l(n_i) \mid a_i \in A_{T, -}, i \in [k]\}$, where $A_{T, +}$ and $A_{T, -}$ are the set of the POS-arcs of $T$ and the set of the NEG-arcs of $T$, respectively.
We interpret that $T$ represents the signed family $\{S^\pm_p \mid p \in \mathcal{P}_T\}$.
We illustrate the TDD representing the signed family $\{\{+1, -2\}, \{+1, -3\}, \{-2, +3\}\}$ in Fig.~\ref{fig:ztdd} for example.
In the figure, a dashed arc ($\dashrightarrow$), a solid single arc ($\rightarrow$), and a solid double arc ($\Rightarrow$) are a ZERO-arc, a POS-arc, and a NEG-arc, respectively.
The TDD in the figure has three directed paths from the root node to $\top$: $1 \rightarrow 2 \Rightarrow \top$, $1 \rightarrow 2 \dashrightarrow 3 \Rightarrow \top$, and $1 \dashrightarrow 2 \Rightarrow 3 \rightarrow \top$, which correspond to $\{+1, -2\}, \{+1, -3\}$, and $\{-2, +3\}$, respectively.

\subsection{Frontier-based search}\label{sec:fbs}
Frontier-based search~\cite{kawahara2017frontier, Knuth11, Sekine1995} is a framework of algorithms that efficiently construct a decision diagram representing the set of subgraphs satisfying given constraints of an input graph.
We explain the general framework of frontier-based search.
Given a graph $G$, let $\mathcal{M}$ be a class of subgraphs we would like to enumerate
(for example, $\mathcal{M}$ is the set of all the $s$-$t$ paths on $G$).
Frontier-based search constructs the ZDD representing the family $\mathcal{M}$ of subgraphs.
By fixing $G$, a subgraph is identified with the edge set the subgraph has,
and thus the ZDD represents the family of edge sets actually.
Non-terminal nodes of ZDDs constructed by frontier-based search have labels $e_1,\ldots,e_m$.
We identify $e_i$ with the integer $i$. We assume that it is determined in advance
which edge in $G$ has which index $i$ of $e_i$.

We directly construct the ZDD in a breadth-first manner.
We first create the root node of the ZDD, make it have label $e_1$,
and then we carry out the following procedure for $i = 1,\ldots,m$.
For each node $n_i$ with label $e_{i}$, we create two nodes, each of which is either a terminal node or a non-terminal node whose label is $e_{i+1}$ (if $i = m$, the candidate is only a terminal node), as the 0-child and the 1-child of $n_i$.

Which node the $x$-arc of a node $n_i$ with label $e_{i}$ points at
is determined by a function, called \textsc{MakeNewNode}, of which
we design the detail according to $\mathcal{M}$, i.e., what subgraphs we want to enumerate.
Here we describe the generalized nature that \textsc{MakeNewNode} must possess.
The node $n_i$ represents the set of the subgraphs, denoted by $\mathcal{G}(n_i)$, corresponding to the set of the directed paths from the root node to $n_i$.
Each subgraph in $\mathcal{G}(n_i)$ contains only edges in $E^{<i}$.
Note that $\mathcal{G}(\top)$ is the desired set of subgraphs represented by the ZDD after the construction finishes.
To decide which node the $x$-arc of $n_i$ points at without traversing the ZDD (under construction),
we make each node $n_i$ have the information $n_i.\mathtt{conf}$, which is shared by all the subgraphs in $\mathcal{G}(n_i)$.
The content of $n_i.\mathtt{conf}$ also depends on $\mathcal{M}$ (for example, in the case of $s$-$t$ paths, we store degrees and components of the subgraphs in $\mathcal{G}(n_i)$ into $n_i.\mathtt{conf}$).
\textsc{MakeNewNode} creates a new node, say $n_\mathrm{new}$, with label $e_{i+1}$ and must behave in the following manner.
\begin{enumerate}
  \item For all edge sets $S^{\leq i} \in \mathcal{G}(n_{\mathrm{new}})$,
  if there is no edge set $S^{>i} \subseteq E^{>i}$ such that
  $S^{\leq i} \cup S^{>i} \in \mathcal{M}$, the function discards $n_\mathrm{new}$ and returns $\bot$ to avoid redundant expansion of nodes. (\emph{pruning})
  \item Otherwise, if $i = m$, the function returns $\top$.
  \item Otherwise, the function calculates ${n_\mathrm{new}}.\mathtt{conf}$ from $n_i.\mathtt{conf}$.
  If there is a node $n_{i+1}$ such that whose label is $e_{i+1}$ and $n_\mathrm{new}.\mathtt{conf} = n_{i+1}.\mathtt{conf}$, the function abandons $n_\mathrm{new}$ and returns $n_{i+1}$. (\emph{node merging})
  If not, the function returns ${n_\mathrm{new}}$.
\end{enumerate}
We make the $x$-arc of $n_i$ point at the node returned by \textsc{MakeNewNode}.

As for $n_i.\mathtt{conf}$, in the case of several kinds of subgraphs such as paths and cycles, it is known that we only have to store states relating to the vertices to which both an edge in $E^{<i}$ and an edge in $E^{\geq i}$ are incident into each node~\cite{Knuth11} (in the case of $s$-$t$ paths, we store degrees and components of such vertices into each node).
The set of the vertices are called the \emph{frontier}.
More precisely, the \emph{$i$-th frontier} is defined as $F_{i} = (\bigcup_{j=1}^{i-1} e_j) \cap (\bigcup_{k=i}^{m} e_k)$.
For convenience, we define $F_{0} = F_{m} = \emptyset$.
States of vertices in $F_{i-1}$ are stored into $n_i.\mathtt{conf}$.
By limiting the domain of the information to the frontier, we can reduce memory consumption and share more nodes, which leads to a more efficient algorithm.

The efficiency of an algorithm based on frontier-based search is often evaluated by the \emph{width of a ZDD} constructed by the algorithm.
The width $W_Z$ of a ZDD $Z$ is defined as $W_Z = \max\{|\mathcal{N}_i| \mid i \in [m]\}$, where $\mathcal{N}_i$ denotes the set of nodes whose labels are $e_i$.
Using $W_Z$, the number of nodes in $Z$ can be written as $|Z| = \mathcal{O}(m W_Z)$ and the time complexity of the algorithm is $\mathcal{O}(\tau |Z|)$, where $\tau$ denotes the time complexity of \textsc{MakeNewNode} for one node.

% 3
\section{Algorithms}\label{sec:algorithms}

\subsection{Overview of the proposed algorithms}\label{sec:overview}
In this section, for a ZDD $\ZA$ and $L \in \mathbb{Z}^{+}$, we propose a novel algorithm to construct the ZDD representing the set of graph partitions such that the partitions are represented by $\ZA$ and each connected component in the partitions has weight at least $L$.
In general, there are two techniques to obtain ZDDs having desired conditions. One is frontier-based search, described in the previous section.
The method proposed by Kawahara et al.~\cite{kawahara2017generating} directly stores the weight of each component into ZDD nodes (as $\mathtt{conf}$) and prunes a node when it is determined that the weight of a component is less than $L$.
However, for two nodes, if the weight of a single component on the one node differs from that on the other node, the two nodes cannot be merged.
Consequently, node merging rarely occurs in Kawahara et al.'s method and thus the size of the resulting ZDD is too large to construct it if the input graph has more than 100 vertices.

The other technique is the usage of the recursive structure of a ZDD. Methods based on the recursive structure are called \emph{apply-like} methods~\cite{bryant1986graph}. For each node $\alpha$ of a ZDD, the nodes and arcs reachable from $\alpha$ compose another ZDD, whose root is $\alpha$. For a ZDD $Z$ and $x \in \{0, 1\}$, let $c_x(Z)$ be the ZDD composed by the nodes and arcs reachable from the $x$-child of the root. For (one or more) ZDDs $F$ (and $G$), an apply-like method constructs a target ZDD by recursively calling itself against $c_0(F)$ and $c_1(F)$ (and $c_0(G)$ and $c_1(G)$). For example, the ZDD representing $F \cap G$ can be computed from $c_0(F) \cap c_0(G)$ and $c_1(F) \cap c_1(G)$. Apply-like methods support various set operations~\cite{bryant1986graph,Knuth11}.

Nakahata et al.~\cite{nakahata2018} developed an algorithm to upperbound the weights of connected components in each partition,
i.e., to construct the ZDD representing the set $\mathcal{A}$ of partitions included in a given ZDD and the weights of all the components in the partitions are at most $H \in \mathbb{Z}^{+}$. Their algorithm first constructs the ZDD $Z_{\mathcal{S}}$ representing the set of forbidden components (described in the introduction) with weight more than $H$ by frontier-based search. Then, the algorithm constructs the ZDD representing $\{ A \in \mathcal{A} \mid \exists S \in \mathcal{S}, A  \supseteq S \}$, written as $Z_{\mathcal{A}}.\mathrm{restrict}(Z_{\mathcal{S}})$, which means the set of all the partitions each of which includes a component in $\mathcal{S}$ as a subgraph, in a way of apply-like methods. Finally, we extract subgraphs not in $Z_{\mathcal{A}}.\mathrm{restrict}(Z_{\mathcal{S}})$ from $Z_{\mathcal{A}}$ by the set difference operation $Z_{\mathcal{A}} \setminus (Z_{\mathcal{A}}.\mathrm{restrict}(Z_{\mathcal{S}}))$~\cite{Minato1993}, which is also an apply-like method.

In our case, lowerbounding the weights of components, it is difficult to compute desired partitions by the above approach because a partition including a forbidden component (i.e., weight less than $L$) \emph{as a subgraph} can be a feasible solution. We want to obtain a partition including a forbidden component \emph{as a connected component}. Although we can perform various set operations by designing apply-like methods, it seems difficult to obtain such partitions by direct set operations.

\begin{figure}[tb]
\begin{center}
\begin{tabular}{cc}
\begin{minipage}[t]{0.45\columnwidth}
\begin{center}
  \includegraphics[width=5cm]{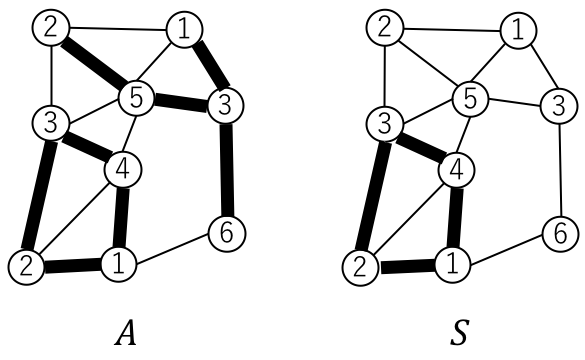}
  \caption{A graph partition $A$ and a connected subgraph $S$. Bold lines are edges contained in the partition or the subgraph. Values in vertices are its weights. $A$ contains $S$ as a connected component. The weight of $S$ is $1 + 2 + 3 + 4 = 10$, and thus, when $L>10$, $A$ does not satisfy the lower bound constraint.}
  \label{fig:subgraph_and_component}
\end{center}
\end{minipage}
\begin{minipage}[t]{0.05\columnwidth}
  \hspace{2mm}
\end{minipage}
\begin{minipage}[t]{0.45\columnwidth}
\begin{center}
  \includegraphics[width=2.2cm]{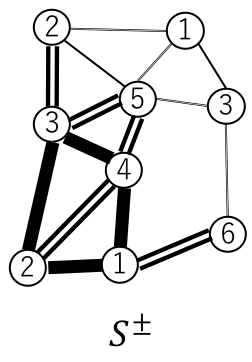}
  \caption{A signed subgraph $S^{\pm}$ with minimal cutset corresponding to $S$ in Fig.~\ref{fig:subgraph_and_component}. Thin single lines, bold single lines, and doubled lines are zero edges, positive edges, and negative edges, respectively.}
  \label{fig:signed_subgraph}
\end{center}
\end{minipage}
\end{tabular}
\end{center}
\end{figure}

Our idea in this paper is to employ the family of signed sets to represent the set of pairs of every forbidden component and its cutset. We use the following observation.
\begin{observation}\label{lem:component}
Let $A$ be a graph partition of $G = (V, E)$ and $S \subseteq E$ be an edge set such that $(\mathrm{dom}(S), S)$ is connected.
The partition $A$ contains $(\mathrm{dom}(S), S)$ as a connected component if and only if both of the following hold.
\begin{enumerate}
  \item $A$ contains all the edges in $S$.
  \item $A$ does not contain any edge $e$ in $E \setminus S$ such that $e$ has at least one vertex in $\mathrm{dom}(S)$.
\end{enumerate}
\end{observation}
Based on Observation~\ref{lem:component}, we associate a signed subgraph $S^{\pm}$ with a connected subgraph $(\mathrm{dom}(S), S)$:
\begin{eqnarray}
S^{\pm} &=& S^{+} \cup S^{-}, \label{eq:spm} \\
S^{+} &=& \{+e \mid e \in S\}, \label{eq:sp} \\
S^{-} &=& \{-e \mid (e \in E \setminus S) \land (e \cap \mathrm{dom}(S) \neq \emptyset)\}. \label{eq:sm}
\end{eqnarray}
$S^{\pm}$ is a signed subgraph such that $\mathrm{abs}(S^{+})$ and $\mathrm{abs}(S^{-})$ are sets of edges satisfying Conditions 1 and 2 in Observation~\ref{lem:component}, respectively.
Note that $\mathrm{abs}(S^{-})$ is a cutset of $G$, that is, removing the edges in $\mathrm{abs}(S^{-})$ separates $G$ into
the connected component $(\mathrm{dom}(\mathrm{abs}(S^{+})), \mathrm{abs}(S^{+}))$ and the rest.
In addition, $\mathrm{abs}(S^{-})$ is minimal among such cutsets.
In this sense, we say that $S^{\pm}$ is a \emph{signed subgraph with minimal cutset for $S$}.

Hereinafter, we call edges in $\mathrm{abs}(S^{+})$ \emph{positive edges}, $\mathrm{abs}(S^{-})$ \emph{negative edges} and the other edges \emph{zero edges}.
Figure~\ref{fig:signed_subgraph} shows $S^{\pm}$ associated with $S$ in Fig.~\ref{fig:subgraph_and_component}.
The partition $A$ in Fig.~\ref{fig:subgraph_and_component} indeed contains all the edges in $\mathrm{abs}(S^{+})$ and does not contain any edges in $\mathrm{abs}(S^{-})$.
For a graph partition $E' \subseteq E$, when the weights of all the connected components of $E'$ is at least $L$, we say that \emph{$E'$ satisfies the lower bound constraint}.
To extract partitions not satisfying the lower bound constraint from an input ZDD, we compute the set of partitions each of which has all the edges in $\mathrm{abs}(S^{+})$ and no edge in $\mathrm{abs}(S^{-})$ for some $S \in \mathcal{S}$.

The overview of the proposed method is as follows.
In the following, let $\mathcal{A}$ be the set of graph partitions represented by the input ZDD and $\mathcal{B}$ be the set of graph partitions each of which belongs to $\mathcal{A}$ and satisfies the lower bound constraint.
\begin{enumerate}
    \item We construct the ZDD $Z_\mathcal{S}$ representing the set $\mathcal{S}$ of forbidden components, where $\mathcal{S}$ is the set of the connected components of $G$ whose weights are less than $L$.
    \item Using $Z_\mathcal{S}$, we construct the TDD $\TS$, where $\Spm$ is a set of signed subgraphs with minimal cutset corresponding to $\mathcal{S}$ by a way of frontier-based search.
    \item Using $T_{\Spm}$, we construct the ZDD $\Zup$, where $\mathcal{S}^\uparrow$ is the set of partitions each of which contains at least one forbidden component in $\mathcal{S}$ as a connected component.
    \item We obtain the ZDD $\ZB$ by the set difference operation $\ZA \setminus \Zup$~\cite{Minato1993}.
\end{enumerate}

In the rest of this section, we describe each step from 1 to 3.

\subsection{Constructing $Z_\mathcal{S}$}\label{sec:zs}
We describe how to construct $Z_\mathcal{S}$, which represents the set $\mathcal{S}$ of forbidden subgraphs whose weights are less than $L$.
In this subsection, we consider only forbidden components with at least one edge.
Note that a component with only one vertex cannot be distinguished by sets of edges because all such subgraphs are represented by the empty edge set.
We show how to deal with components having only one vertex in Sec.~\ref{sec:superset}.

We can construct $Z_\mathcal{S}$ using frontier-based search. Due to the page restriction, we describe a brief overview.
To construct $Z_\mathcal{S}$, in the frontier-based search, it suffices to ensure that every enumerated subgraph has only one connected component and its weight is less than $L$.
The former can be dealt by storing the connectivity of the vertices in the frontier as $\mathtt{comp}$~\cite{kawahara2017frontier}.
The latter can be checked by managing the total weight of vertices such that at least one edge is incident to as $\mathtt{weight}$.

Let us analyze the width of $Z_\mathcal{S}$.
For nodes with the same label, there are $\mathcal{O}(B_f)$ different states for $\mathtt{comp}$~\cite{kawahara2017generating}, where, for $k \in \mathbb{Z}^{+}$, $B_k$ is the $k$-th Bell number and $f = \max\{|F_i| \mid i \in [m]\}$.
As for $\mathtt{weight}$, when $\mathtt{weight}$ exceeds $L$, we can immediately conclude that the subgraphs whose weights are less than $L$ are generated no more.
If we prune such cases, there are $\mathcal{O}(L)$ different states for $\mathtt{weight}$.
As a result, we can obtain the following lemma on the width of $Z_\mathcal{S}$.
\begin{lemma}
  The width of $\ZS$ is $\mathcal{O}(B_f L)$, where $f = \max\{|F_i| \mid i \in [m]\}$.
\end{lemma}

\subsection{Constructing $\TS$}\label{sec:ts}
In this subsection, we propose an algorithm to construct $\TS$.
First, we show how to construct the TDD representing the set of all the signed subgraphs with minimal cutset, including a disconnected one.
Next, we describe the method to construct $\TS$ using $Z_\mathcal{S}$.

Let $S^{\pm} = S^{+} \cup S^{-}$ be a signed subgraph.
Our algorithm uses the following observation on signed subgraphs with minimal cutset.
\begin{observation}\label{lem:bordered_subgraph}
  A signed subgraph $S^{\pm}$ is a signed subgraph with minimal cutset if and only if the following two conditions hold:
  \begin{enumerate}
    \item For all $v \in V$, at most one of a zero edge or a positive edge is incident to $v$.
    \item For all the negative edges $\{u, v\}$, a positive edge is incident to at least one of $u$ and $v$.
  \end{enumerate}
\end{observation}
Conditions 1 and 2 in Observation~\ref{lem:bordered_subgraph} ensure that $\mathrm{abs}(S^{-})$ is a cutset such that removing it leaves the connected component whose edge set is $\mathrm{abs}(S^{+})$ and the minimality of $\mathrm{abs}(S^{-})$.
This shows the correctness of the observation.
We design an algorithm based on frontier-based search to construct a TDD representing the set of all the signed subgraphs satisfying Conditions 1 and 2 in Observation~\ref{lem:bordered_subgraph}.

First, we consider Condition 1.
To ensure Condition 1, we store an array $\mathtt{colors} : V \rightarrow 2^{\{0, +, -\}}$ into each TDD node.
For all $v \in F_{i - 1}$, we manage $n_i.\mathtt{colors}[v]$ so that it is equal to the set of types of edges incident to $v$.
For example, if a zero edge and a positive edge are incident to $v$ and no negative edges are, $\mathtt{colors}[v]$ must be $\{0, +\}$.
We can prune the case such that Condition 1 is violated using $\mathtt{colors}$, which ensures Condition 1.

Next, we consider Condition 2.
Let $\{u, v\}$ be a negative edge.
When $u$ and $v$ leave the frontier at the same time, we check if Condition 2 is satisfied from $\mathtt{colors}[u]$ and $\mathtt{colors}[v]$ and, if not, we prune the case.
When one of $u$ or $v$ leaves the frontier (without loss of generality, we assume the vertex is $u$), if no positive edges are incident to $u$, at least one positive edge must be incident to $v$ later.
To deal with this situation, we store an array $\mathtt{reserved}: V \rightarrow \{0, 1\}$ into each TDD node.
For all $v \in F_{i - 1}$, we manage $\mathtt{reserved}[v]$ so that $\mathtt{reserved}[v] = 1$ if and only if at least one positive edge must be incident to $v$ later.
We can prune the cases such that $v \in V$ is leaving the frontier and both $\mathtt{reserved}[v] = 1$ and $+ \notin \mathtt{colors}[v]$ hold, which violate Condition 2.
We show \textsc{MakeNewNode} function and its subroutine \textsc{Reserve} in Algorithms~\ref{alg:make_new_node} and \ref{alg:reserve} in Appendix~\ref{app:code}, respectively.

We give the following lemma on the width of a ZDD constructed by  Algorithms~\ref{alg:make_new_node} and \ref{alg:reserve}.
\begin{lemma}\label{lem:width_bordered_subgraph}
  The width $W_T$ of a ZDD constructed by Algorithms~\ref{alg:make_new_node} and \ref{alg:reserve} is $W_T = \mathcal{O}(6^f)$.
\end{lemma}
\begin{proof}
  We analyze the number of different non-terminal nodes which are returned by \textsc{MakeNewNode} function and have the label $e_i$.
  To this end, we analyze the number of a pair $(\mathtt{colors}[w], \mathtt{reserved}[w])$ for each $w \in F_{i - 1}$.
  Because of Lines~4--5 in \textsc{MakeNewNode},
  $+$ and $0$ are never in $\mathtt{colors}[w]$ together.
  In addition, $\mathtt{colors}[w]$ is never empty because, when \textsc{MakeNewNode} returns a non-terminal node, there are at least one processed edge incident to $w$ and its type has been added into $\mathtt{colors}[w]$ in Line~16.
  Therefore, there are at most five different states for $\mathtt{colors}[w]$: $\{0\}, \{-\}, \{+\}, \{0, -\}$, and $\{-, +\}$.
  As for $\mathtt{reserved}[w]$, it may be 1 only when $\mathtt{colors}[w] = \{-\}$ because of Lines 3--4 in \textsc{Reserve}.
  Thus, there are at most six different states for $(\mathtt{colors}[w], \mathtt{reserved}[w])$.
  There are at most $f$ vertices in the frontier, and therefore $W_T = \mathcal{O}(6^f)$.
\end{proof}

Next, we show how to construct $\TS$ using $Z_\mathcal{S}$.
We can achieve this goal using \emph{subsetting} technique~\cite{iwashita2013efficient} with Algorithms~\ref{alg:make_new_node} and \ref{alg:reserve}.
Subsetting technique is a framework to construct a decision diagram corresponding to another decision diagram.
We ensure that, for all $S^{\pm} = S^{+} \cup S^{-} \in \Spm$, there exists $S \in \mathcal{S}$ such that $\mathrm{abs}(S^{+}) = S$ in the construction of $\TS$ using subsetting technique.

\subsection{Constructing $\Zup$}\label{sec:superset}
In this section, we show how to construct $\Zup$ and how to deal with forbidden components consisting only of one vertex whose weight is less than $L$, which was left as a problem in Sec.~\ref{sec:zs}.
From Observation~\ref{lem:component} and Eqs.~\eqref{eq:spm}--\eqref{eq:sm}, $\mathcal{S}^{\uparrow}$ can be written as
\begin{equation}
  \mathcal{S}^{\uparrow} = \{E' \subseteq E
    \mid \exists S^{\pm} \in \Spm,
    (\forall +e \in S^{\pm}, e \in E') \land
    (\forall -e \in S^{\pm}, e \notin E')\}.
\end{equation}
Using $\TS$, we can construct $\ZS$ by the algorithm of Suzuki et al.~\cite{suzuki2018enumeration}.

Finally, we show how to deal with a graph partition containing a single vertex $v$ such that $p(v) < L$ as a connected component, i.e., a partition has an isolated vertex with small weight.
Let $\mathcal{F}_v$ be the set of graph partitions containing $(\{v\}, \emptyset)$ as a connected component.
A graph partition $E' \subseteq E$ belongs to $\mathcal{F}_v$ if and only if $E'$ does not contain any edge incident to $v$.
Using this, we can construct the ZDD $Z_v$ representing $\mathcal{F}_v$ in $\mathcal{O}(m)$ time.
For each $v \in V$ such that $p(v) < L$, we construct $Z_v$ and update $\Zup \gets \Zup \cup Z_v$.
In this way, we can deal with all the graph partitions containing a connected component whose weight is less than $L$.
We show an example of execution of the whole algorithm in Appendix~\ref{app:example}.

\section{Experimental results}\label{sec:experiment}
We conducted computational experiments to evaluate the proposed algorithm and to compare it with the existing state-of-the-art algorithm of Kawahara et al~\cite{kawahara2017generating}.
We used a machine with an Intel Xeon Processor E5-2690v2 (3.00 GHz) CPU and a 64 GB memory (Oracle Linux 6) for the experiments.
We have implemented the algorithms in C++ and compiled them by g++ with the \texttt{-O3} optimization option.
In the implementation, we used the \texttt{TdZdd} library~\cite{iwashita2013efficient} and the \texttt{SAPPORO\_BDD} library.\footnote{Although the \texttt{SAPPORO\_BDD} library is not released officially, you can see the code in \url{https://github.com/takemaru/graphillion/tree/master/src/SAPPOROBDD}.}
The timeout is set to be an hour.

\begin{table}[tb]
\caption{Summary of input graphs and input graph partitions.}
\label{tab:map_input}
\hbox to\hsize{\hfil
\scalebox{0.9}{
\begin{tabular}{l|rrr|rr|rr|rr}\hline
       &       &        &
& \multicolumn{2}{c}{Induced partition}
& \multicolumn{2}{|c}{Forest}
& \multicolumn{2}{|c}{Rooted forest} \\

Name   & $n$   & $m$    & $k$   & $|Z_{\mathcal{A}}|$ & $ |\mathcal{A}|$
& $|Z_{\mathcal{A}}|$ & $ |\mathcal{A}|$ & $|Z_{\mathcal{A}}|$ & $ |\mathcal{A}|$ \\\hline

$G_1$ (Gumma) & 37	  & 80    & 4
& 10236   & $1.25 \times 10^{8}$
& 26361   & $1.01 \times 10^{19}$
& 8957    & $1.66 \times 10^{16}$ \\

$G_2$ (Ibaraki) & 44	  & 95    & 7
& 17107   & $6.38 \times 10^{13}$
& 15553   & $6.14 \times 10^{23}$
& 3238    & $1.94 \times 10^{19}$ \\

$G_3$ (Chiba) & 60	  & 134   & 14
& 301946  & $6.69 \times 10^{22}$
& 213773  & $4.86 \times 10^{33}$
& 15741   & $5.04 \times 10^{25}$ \\

$G_4$ (Aichi) & 69	  & 173   & 17
& 1598213  & $9.26 \times 10^{29}$
& 879361   & $1.78 \times 10^{42}$
& 43465    & $3.10 \times 10^{30}$ \\

$G_5$ (Nagano) & 77   & 185   & 5
& 13203   & $2.77 \times 10^{17}$
& 44804   & $2.95 \times 10^{43}$
& 26476   & $7.66 \times 10^{39}$\\

\hline
\end{tabular}
}\hfil}
\end{table}

% \begin{landscape}
\begin{sidewaystable}
% \begin{table}[tb]
\caption{Experimental results for three types of input graph partitions.}
\label{tab:map_result}
\hbox to\hsize{\hfil
\scalebox{0.9}{
\begin{tabular}{l|rr|rrrr|rrrr|rrrr}\hline
&          &
& \multicolumn{4}{c}{Induced partition}
& \multicolumn{4}{|c}{Forest}
& \multicolumn{4}{|c}{Rooted forest} \\

& $r$       & $L(r, k)$
& Alg.~N  & Alg.~K & $|Z_{\mathcal{B}}|$ & $ |\mathcal{B}|$
& Alg.~N  & Alg.~K & $|Z_{\mathcal{B}}|$ & $ |\mathcal{B}|$
& Alg.~N  & Alg.~K & $|Z_{\mathcal{B}}|$ & $ |\mathcal{B}|$ \\\hline

\multirow{5}{*}{$G_1$}
&1.1     & 458947
& \textbf{4.22}   & 12.07 & 4912  & $1.74 \times 10^{4}$
& \textbf{4.03}   & 50.84 & 29502 & $8.24 \times 10^{12}$
& \textbf{3.95}   & 14.96 & 17920 & $3.52 \times 10^{11}$
\\

&1.2     & 429016
& \textbf{2.06}   & 10.50   & 3500  & $5.40 \times 10^{4}$
& \textbf{2.04}   & 47.30   & 21364 & $3.10 \times 10^{13}$
& \textbf{2.02}   & 13.34   & 6331  & $1.68 \times 10^{12}$
\\

&1.3     & 402750
& \textbf{1.15}   & 7.49    & 2986  & $9.02 \times 10^{4}$
& \textbf{1.18}   & 36.10   & 18113 & $7.42 \times 10^{13}$
& \textbf{1.17}   & 10.54   & 4655  & $4.44 \times 10^{12}$
\\

&1.4     & 379514
& \textbf{0.99}   & 5.72    & 3115  & $2.52 \times 10^{5}$
& \textbf{1.03}   & 24.41   & 20605 & $3.84 \times 10^{14}$
& \textbf{1.03}   & 6.97    & 7677  & $3.18 \times 10^{13}$
\\

&1.5     & 358813
& \textbf{0.90}   & 5.12    & 3562  & $2.99 \times 10^{5}$
& \textbf{0.89}   & 23.29   & 20367 & $7.19 \times 10^{14}$
& \textbf{0.88}   & 6.52    & 6719  & $6.17 \times 10^{13}$
\\\hline

\multirow{5}{*}{$G_2$}
&1.1     & 383928
& \textbf{3.70}   & 29.48   & 27927   & $1.91 \times 10^{6}$
& \textbf{3.60}   & 35.28   & 47461   & $2.56 \times 10^{13}$
& 3.53            & \textbf{2.19}    & 391     & $4.32 \times 10^{6}$
\\

&1.2     & 355836
& \textbf{3.03}   & 23.03   & 83053   & $1.25 \times 10^{8}$
& \textbf{2.92}   & 25.59   & 143455  & $2.11 \times 10^{15}$
& 2.95            & \textbf{1.81}    & 3103    & $3.72 \times 10^{9}$
\\

&1.3     & 331574
& \textbf{1.73}   & 16.25   & 92334   & $1.02 \times 10^{9}$
& \textbf{1.70}   & 18.09   & 154449  & $1.41 \times 10^{16}$
& \textbf{1.60}   & 1.74    & 5861    & $1.36 \times 10^{11}$
\\

&1.4     & 310410
& \textbf{1.21}   & 12.45   & 105507  & $4.54 \times 10^{9}$
& \textbf{1.30}   & 14.03   & 179186  & $1.02 \times 10^{17}$
& \textbf{1.28}   & 1.55    & 5710    & $1.54 \times 10^{12}$
\\

&1.5     & 291785
& \textbf{0.73}   & 8.88    & 98231   & $1.25 \times 10^{10}$
& \textbf{0.74}   & 9.38    & 149403  & $3.06 \times 10^{17}$
& \textbf{0.70}   & 1.21    & 5855    & $6.74 \times 10^{12}$
\\\hline

\multirow{5}{*}{$G_3$}
&1.1     & 377742
& \textbf{83.76}  & 1008.11 & 0       & 0
& \textbf{77.19}  & 811.03  & 0       & 0
& 78.68           & \textbf{66.96}   & 0       & 0
\\

&1.2     & 348159
& \textbf{32.87}  & 852.47  & 6641    & $2.32 \times 10^{5}$
& \textbf{27.12}  & 657.89  & 17252   & $1.34 \times 10^{13}$
& \textbf{27.27}  & 89.75   & 0       & 0
\\

&1.3     & 322874
& \textbf{23.33}  & 626.94  & 261978  & $3.12 \times 10^{10}$
& \textbf{20.87}  & 452.10  & 768876  & $1.53 \times 10^{19}$
& \textbf{36.20}  & 36.30   & 0       & 0
\\

&1.4     & 301013
& \textbf{12.08}  & 386.91  & 328581  & $4.92 \times 10^{11}$
& \textbf{10.88}  & 266.19  & 917102  & $3.23 \times 10^{20}$
& \textbf{9.70}   & 22.14   & 0       & 0
\\

&1.5     & 281924
& \textbf{10.81}  & 315.40  & 405816  & $3.02 \times 10^{12}$
& \textbf{9.29}   & 205.90  & 1062331 & $9.94 \times 10^{20}$
& \textbf{7.64}   & 19.44   & 606     & $2.88 \times 10^{10}$
\\\hline

\multirow{5}{*}{$G_4$}
&1.1     & 402370
& \textbf{155.05}   & OOM     & 190520    & $1.54 \times 10^{10}$
& \textbf{64.12}    & 1032.53 & 374111    & $5.43 \times 10^{18}$
& 51.95             & \textbf{0.65}    & 0         & 0
\\

&1.2     & 370499
& \textbf{86.91}    & 628.93  & 739356    & $1.98 \times 10^{14}$
& \textbf{24.09}    & 317.44  & 1374522   & $1.41 \times 10^{23}$
& 20.82             & \textbf{0.96}    & 0         & 0
\\

&1.3     & 343307
& \textbf{125.06}   & 408.97  & 1148330   & $1.98 \times 10^{16}$
& \textbf{14.83}    & 190.25  & 2005760   & $7.27 \times 10^{24}$
& 11.69             & \textbf{1.48}    & 0         & 0
\\

&1.4     & 319833
& \textbf{108.25}   & 281.81  & 1465722   & $6.32 \times 10^{17}$
& \textbf{12.18}    & 134.15  & 2495000   & $1.87 \times 10^{26}$
& 8.31              & \textbf{3.09}    & 5645      & $2.19 \times 10^{11}$
\\

&1.5     & 299363
& \textbf{29.13}    & 190.59  & 1761682   & $1.65 \times 10^{19}$
& \textbf{9.60}     & 85.84   & 2434632   & $4.02 \times 10^{27}$
& 5.55              & \textbf{3.46}    & 15587     & $9.56 \times 10^{14}$
\\\hline

\multirow{5}{*}{$G_5$}
&1.1     & 388844
& $>$ 1 h   & OOM       & - & -
& $>$ 1 h   & OOM       & - & -
& $>$ 1 h   & \textbf{$<$ 0.01}  & 0 & 0
\\

&1.2     & 362027
& $>$ 1 h   & OOM     & - & -
& $>$ 1 h   & OOM     & - & -
& $>$ 1 h   & \textbf{$<$ 0.01}  & 0 & 0
\\

&1.3     & 338670
& OOM       & OOM       & - & -
& $>$ 1 h   & OOM       & - & -
& $>$ 1 h   & \textbf{$<$ 0.01}  & 0 & 0
\\

&1.4     & 318145
& OOM       & OOM     & - & -
& OOM       & OOM     & - & -
& OOM       & \textbf{$<$ 0.01}  & 0 & 0
\\

&1.5     & 299965
& OOM       & \textbf{1960.28} & 393178  & $9.20 \times 10^{13}$
& OOM       & OOM     & - & -
& OOM       & \textbf{$<$ 0.01}  & 0 & 0
\\
\hline
\end{tabular}
}\hfil}
\end{sidewaystable}
% \end{table}
% \end{landscape}

We used graphs representing some prefectures in Japan for the input graphs.
The vertices represent cities and there is an edge between two cities if and only if they have the common border.
The weight of a vertex represents the number of residents living in the city represented by the vertex.
As for the input ZDD $\ZA$, we adopted three types of graph partitions: graph partitions such that each connected component is an induced subgraph~\cite{kawahara2017generating}, which we call \emph{induced partition}, forests, and rooted forests.
There is a one-to-one correspondence between induced partitions and partitions of the vertex set.
A rooted forest is a forest such that each tree in the forest has exactly one specified vertex.
We chose special vertices for each graph randomly.
A summary of input graphs and input graph partitions is in Tab.~\ref{tab:map_input}.
In the table, we show graph names and the prefecture represented by the graph, the number of vertices ($n$), edges ($m$) and connected components ($k$) in graph partitions.
The groups of columns ``Induced partition'', ``Forest'', and ``Rooted forest'' indicate the types of input graph partitions.
Inside each of them, we show the size (the number of non-terminal nodes) of $\ZA$ and the cardinality of $\mathcal{A}$.

The lower bounds of weights are determined as follows.
Let $k$ be the number of connected components in a graph partition and $r$ be the maximum ratio of the weights of two connected components in the graph partition.
From $k$ and $r$, we can derive the necessary condition that the weight of every connected component must be at least $L(k, r) = P / (r(k-1)+1)$, where $P = \sum_{v \in V} p(v)$~\cite{kawahara2017generating}.
We used $L(k, r)$ as the lower bound of weights in the experiment.
For each graph, we run the algorithms in $r = 1.1, 1.2, 1.3, 1.4$, and $1.5$.

We show the experimental results in Tab.~\ref{tab:map_result}.
In the table, we show the graph name, the value of $r$ and $L(k, r)$, and the execution time of \emph{Alg.~N}, the proposed algorithm, and \emph{Alg.~K}, the algorithm of Kawahara et al.
The size of $\ZB$ and the cardinality of $\mathcal{B}$ are also shown.
``OOM'' means \emph{out of memory} and ``-'' means both algorithms failed to construct the ZDD (due to timeout or out of memory).
We marked the values of the time of the algorithm which finished faster as bold.

First, we analyze the results for induced partitions.
For the input graphs from $G_1$ to $G_4$, both Alg.~N and Alg.~K succeeded in constructing $\ZB$, except when $r = 1.1$ in $G_4$ for Alg.~K.
In cases where both algorithms succeeded in constructing $\ZB$, the time for Alg.~N to construct the ZDD is 2--32 times shorter than that for Alg.~K.
In addition, Alg.~N succeeded in constructing the ZDD when $r = 1.1$ in $G_4$, where Alg.~K failed to construct the ZDD because of out of memory.
These results show the efficiency of our algorithm.
In contrast, for $G_5$, although both algorithms failed to construct the ZDD when $r = 1.1, 1.2, 1.3$ and $1.4$, only Alg.~K succeeded when $r = 1.5$.
In this case, the size of the ZDD constructed by Alg.~N did stay in the limitation of memory while, in our algorithm, the size of $\Zup$ exceeded the limitation of memory.

Second, we investigate the results for forests.
Both Alg.~N and Alg.~K succeeded in constructing $\ZB$ for the input graph from $G_1$ to $G_4$.
In all those cases, Alg.~N was faster than Alg.~K.
Comparing the results with those of induced partitions, we found that the execution time of Alg.~K depends on the input partitions more than Alg.~N does.
For example, for $G_1$, while the execution time of Alg.~N is almost irrelevant to the types of input ZDDs, that of Alg.~K differ up to about five times.
This is because the efficiency of Alg.~K strongly depends on the sizes of input ZDDs.
This makes the sizes of output ZDDs constructed by Alg.~K large, which implies the increase in the execution time of Alg.~K.
In contrast, the execution time of Alg.~N does not depend on the sizes of input ZDDs in many cases because Alg.~N uses the input ZDD only in the set difference operation, which is executed in the last of the algorithm (by the existing apply-like method).
As we show later, the bottleneck of Alg.~N is the construction of $\Zup$.
Therefore, in many cases, the sizes of input ZDDs do not change the execution time of Alg.~N.

Third, we examine the results when the input graph partitions are rooted forests.
There are 13 cases such that Alg.~K was faster than Alg.~N.
In the cases, the sizes of input ZDDs and output ZDDs are small, that is, thousands, or even zero.
These results show that Alg.~K tends to be faster when the sizes of input ZDDs and output ZDDs are small.

In order to assess the efficiency of our algorithm in each step, we show detailed experimental results for $G_3$ and $G_4$ when the input graph partitions are induced partitions in Tab.~\ref{tab:map_result_detail}.
In the table, we show the time to construct decision diagrams, the size of decision diagrams, and the cardinality of the family represented by ZDDs.
The cardinality of $S^\pm$ is omitted because it is equal to that of $\mathcal{S}$.
The size and cardinality for $\ZA \setminus \Zup$ are also omitted because they are the same as $|\ZB|$ and $|\mathcal{B}|$, which are shown in Tab.~\ref{tab:map_result}.
For both $G_3$ and $G_4$, the time to construct $Z_\mathcal{S}$ and $\TS$ are within one or two seconds.
The most time-consuming parts are the construction of $\Zup$ in $G_3$ and $\Zup$ or $\ZA \setminus \Zup$ in $G_4$.
The set difference operation in $G_4$ took a lot of time because the sizes of $\ZA$ and $\Zup$ are large, that is, more than a hundred.
The reason why the construction of $\Zup$ takes a lot of time is the increase in the sizes of decision diagrams.
While the size of $\TS$ is only 2--7 times larger than that of $\ZS$, that of $\Zup$ is about 10--276 times larger than that of $\TS$.
This also made the execution of the algorithm in $G_5$ impossible.

\begin{table}[tb]
\caption{Detailed experimental results of the proposed algorithm for $G_3$ (Chiba) and $G_4$ (Aichi) when the input graph partitions are induced partitions.}
\label{tab:map_result_detail}
\hbox to\hsize{\hfil
\scalebox{0.9}{
\begin{tabular}{r|r|rrr|rr|rrr|r}\hline

&    & \multicolumn{3}{|c}{$\ZS$}
     & \multicolumn{2}{|c}{$\TS$}
     & \multicolumn{3}{|c}{$\Zup$}
     & \multicolumn{1}{|c}{$\ZA \setminus \Zup$} \\

&$r$ & time     & node      & card
     & time     & node
     & time     & node      & card
     & time   \\\hline

\multirow{5}{*}{$G_3$}
&1.1 & 1.90     & 54745     & $4.24 \times 10^{8}$
     & 0.93     & 99057
     & 75.88    & 2117874   & $2.17532 \times 10^{40}$
     & 5.05 \\
&1.2  & 1.01     & 39845     & $1.67 \times 10^{8}$
     & 0.69     & 75581
     & 27.94    & 977840    & $2.17528 \times 10^{40}$
     & 3.23 \\
&1.3  & 0.58     & 31030     & $6.62 \times 10^{7}$
     & 0.51     & 60034
     & 18.83    & 814538    & $2.17498 \times 10^{40}$
     & 3.41 \\
&1.4  & 0.34     & 24066     & $3.30 \times 10^{7}$
     & 0.38     & 48818
     & 8.49     & 490753    & $2.17490 \times 10^{40}$
     & 2.87 \\
&1.5  & 0.25     & 19877     & $1.42 \times 10^{7}$
     & 0.34     & 40340
     & 7.23     & 410152    & $2.17486 \times 10^{40}$
     & 2.99 \\\hline

\multirow{5}{*}{$G_4$}
&1.1  & 0.02     & 2376      & $2.09 \times 10^{4}$
     & 0.32     & 11109
     & 80.03    & 3074734   & $1.19200 \times 10^{52}$
     & 74.68 \\
&1.2  & 0.01     & 1686      & $1.03 \times 10^{4}$
     & 0.20     & 8511
     & 22.24    & 1205320   & $1.19174 \times 10^{52}$
     & 64.46 \\
&1.3  & 0.01     & 1235      & $6.11 \times 10^{3}$
     & 0.17     & 6935
     & 11.51    & 692798    & $1.19170 \times 10^{52}$
     & 113.37 \\
&1.4  & $< 0.01$ & 961       & $3.67 \times 10^{3}$
     & 0.14     & 5808
     & 8.30     & 529214    & $1.19164 \times 10^{52}$
     & 99.81 \\
&1.5  & $< 0.01$ & 756       & $2.67 \times 10^{3}$
     & 0.13     & 4930
     & 5.30     & 348832    & $1.19153 \times 10^{52}$
     & 23.70 \\
\hline
\end{tabular}
}\hfil}
\end{table}

\section{Conclusion}\label{sec:conclusion}
In this paper, we have proposed an algorithm to construct a ZDD representing all the graph partitions such that all the weights of its connected components are at least a given value.
As shown in the experimental results, the proposed algorithm has succeeded in constructing a ZDD representing a set of more than $10^{12}$ graph partitions in ten seconds, which is 30 times faster than the existing state-of-the-art algorithm.
Future work is devising a more memory efficient algorithm that enables us to deal with larger graphs, that is, graphs with hundreds of vertices.
It is also important to seek for efficient algorithms to deal with other constraints on weights such that the ratio of the maximum and the minimum of weights is at most a specified value.

\appendix

% Appendix A
\section{Example of execution of the whole algorithm}\label{app:example}
We show an example of execution of the whole algorithm.
Let an input graph $G$ be the cycle of four vertices, as shown in Fig.~\ref{fig:cycle4}.
In the figure, an integer in a vertex represents its weight.
Let $\mathcal{A}$ be the set of partitions of $G$ which has exactly two components.
$\mathcal{A}$ consists of six partitions as shown in Fig.~\ref{fig:graph_sets}.
Let us extract graph partitions from $\mathcal{A}$ each of whose components have at least weight $L = 3$.
First, we enumerate connected subgraphs in $G$ whose weights are less than $L$ and have at least two vertices.
$G$ contains only one such subgraph, $\{e_1\}$.
Thus $\mathcal{S}$ consists of only the subgraph.
Second, we obtain the set $\Spm$ of signed subgraphs each of which is a signed subgraph with minimal cutset for a subgraph in $\mathcal{S}$.
Now $\Spm$ is $\{\{+e_1, -e_2, -e_3\}\}$.
Third, we calculate the set $\Sup$ of graph partitions each of which contains a component whose weight is less than $L$ as a connected subgraph.
There is two partitions containing $\{e_1\}$ as a connected component: $\{e_1\}$ and $\{e_1, e_4\}$.
In addition, we consider subgraphs with only one vertex whose weight is less than $L$: $(\{v_1\}, \emptyset)$ and $(\{v_2\}, \emptyset)$.
Adding partitions which have $v_1$ or $v_2$ as an isolated vertex, we obtain $\Sup$ and it consists of eight partitions, as shown in the figure.
Finally, we remove the graph partitions in $\Sup$ from $\mathcal{A}$ and obtain the solutions: $\{e_1, e_2\}, \{e_1, e_3\}$, and $\{e_2, e_3\}$.

\begin{figure}[tb]
  \centering
  \includegraphics[width=2cm]{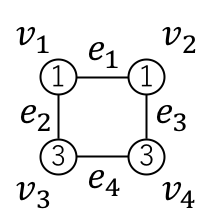}
  \caption{An example input graph.}
  \label{fig:cycle4}
\end{figure}

\begin{figure}[tb]
  \centering
  \includegraphics[width=12cm]{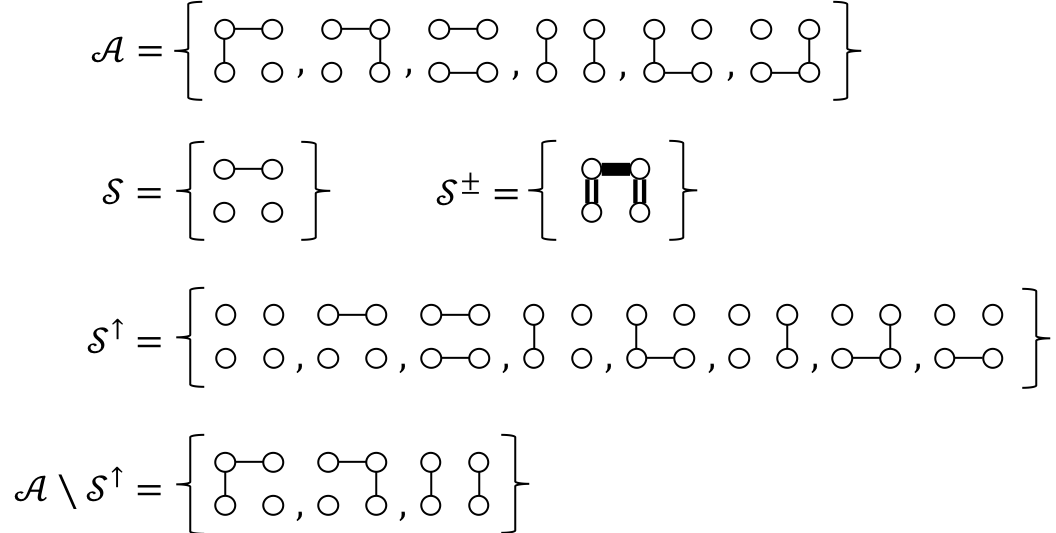}
  \caption{Graph sets in the execution of the algorithm.}
  \label{fig:graph_sets}
\end{figure}

% Appendix B
\section{Pseudocode}\label{app:code}
We show pseudocode in Algorithms~\ref{alg:make_new_node} and \ref{alg:reserve} referred to from the main sections.

\begin{algorithm}[tb]
\caption{$\textsc{MakeNewNode}(n_i, i, s)$ for constructing a TDD representing the set of signed subgraphs with minimal cutset.}
\label{alg:make_new_node}
\DontPrintSemicolon
\setstretch{0.9}
\tcp{This function returns $s (\in \{0, +, -\})$-child of $n_i$ whose label is $e_i$.}

Let $e_i = \{u, v\}$.\;
Copy $n_i$ to $n'_i$.\;

\ForEach{$x \in \{u, v\}$} {
    \tcp{violates Condition 1 in Observation~\ref{lem:bordered_subgraph}}
    \lIf{$0 \in n'_i.\mathtt{colors}[x]$ {\bf and} $s = +$}{\Return $\bot$}
    \lIf{$+ \in n'_i.\mathtt{colors}[x]$ {\bf and} $s = 0$}{\Return $\bot$}
    \If{$n'_i.\mathtt{colors}[x] = \{-\}$ {\bf and} $s = 0$}{
        \tcp{Reserve the vertices in the frontier which are connected to $x$ by the processed edges.}
        $n'_i \gets {\textsc{Reserve}}(n'_i, N(E^{<i}, x) \cap (F_{i - 1} \cup F_{i}))$\;
        \lIf{$n'_i = \bot$}{\Return $\bot$}
    }
    \If{$0 \in n'_i.\mathtt{colors}[x]$ {\bf and} $s = -$} {
        \tcp{There is a zero edge incident to $x$, which is an endpoint of $e_i$. Therefore, we reserve the other endpoint of $e_i$.}
        $n'_i \gets {\textsc{Reserve}}(n'_i, e_i \setminus \{x\})$\;
        \lIf{$n'_i = \bot$}{\Return $\bot$}
    }
    \If{$n'_i.\mathtt{reserved}[x] = 1$ {\bf and} $s = 0$} {
        \tcp{Since $x$ has been already reserved, if $e_i$ (zero edge) becomes incident to $x$, positive edges cannot become incident to $x$. This violates Condition 1 in Observation~\ref{lem:bordered_subgraph}.}
        \Return $\bot$\;
    }
    \If{$n'_i.\mathtt{reserved}[x] = 1$ {\bf and} $s = +$} {
        $n'_i.\mathtt{reserved}[x] \gets 0$ \tcp*{The reservation is archived.}
    }
    $n'_i.\mathtt{colors}[x] \gets n'_i.\mathtt{colors}[x] \cup \{s\}$\;
}

\ForEach{$x \in \{u, v\}$} {
    \If{$x \notin F_{i}$} {
        \tcp{$x$ is leaving the frontier.}
        \If{$n'_i.\mathtt{reserved}[x] = 1$ {\bf and} $+ \notin n'_i.\mathtt{colors}[x]$} {
            \tcp{Although $x$ is reserved, no positive edges are incident to $x$.}
            \Return $\bot$\;
        }
        \If{$n'_i.\mathtt{colors}[x] = \{-\}$} {
            \tcp{Reserve the vertices in the frontier which are connected to $x$ by the processed edges.}
            $n'_i \gets {\textsc{Reserve}}(n'_i, N(E^{\leq i}, x) \cap (F_{i - 1} \cup F_{i}))$\;
            \lIf{$n'_i = \bot$}{\Return $\bot$}
        }
        \tcp{Delete the information about the vertices leaving the frontier.}
        $n'_i.\mathtt{colors}[x] \gets \{\}$\;
        $n'_i.\mathtt{reserved}[x] \gets 0$\;
    }
}

\If{$i = m$} {
    \Return $\top$ \tcp*{All the constraints are satisfied.}
}

\Return $n'_i$\;

\end{algorithm}

\begin{algorithm}[tb]
\caption{${\textsc{Reserve}(n', X)}$}
\label{alg:reserve}
\DontPrintSemicolon
\setstretch{0.9}
\tcp{This function reserves the vertices in $X \subseteq V$ in a TDD node $n'$ and returns the node $n''$ who has an updated state from $n'$.}

Copy $n'$ to $n''$.\;

\For{$x \in X$} {
    \tcp{We cannot reserve $x$ if there is a zero edge incident to $x$.}
    \lIf{$0 \in n''.\mathtt{colors}[x]$} {\Return $\bot$}
    \tcp{Reserve $x$ if there are no positive edges incident to $x$.}
    \lIf{$+ \notin n''.\mathtt{colors}[x]$} {$n''.\mathtt{reserved}[x] \gets 1$}
}

\Return $n''$\;

\end{algorithm}

\bibliographystyle{plain}
% \bibliography{reference}

\end{document}